\documentclass[12pt, draftclsnofoot, onecolumn]{IEEEtran}
\usepackage{amsmath}
\usepackage{amssymb}
\usepackage{graphicx}
\usepackage{color}
\usepackage[latin1]{inputenc}

\newtheorem{lemma}{Lemma}
\newtheorem{theorem}{Theorem}
\newtheorem{corollary}{Corollary}

\def\snr {\mbox{\scriptsize\sf SIR}}
\def\snr {\mbox{\scriptsize\sf SNR}}
\def\sinr {\mbox{\scriptsize\sf SINR}}
\newcommand{\PP}{\mathbb{P}}
\newcommand{\E}{\mathbb{E}}

\newcommand{\dd}{{\rm d}}

\def\chr#1{#1}

\begin{document}

\title{A Tractable Approach to Coverage and Rate in Cellular Networks}
\author{Jeffrey G. Andrews, Fran\c{c}ois Baccelli, and Radha Krishna Ganti
\thanks{F. Baccelli is with Ecole Normale Superieure (ENS) and INRIA in Paris, France, J. G. Andrews and R. K. Ganti are with the Dept. of ECE, at the University of Texas at Austin.  The contact email is jandrews@ece.utexas.edu.
Date revised: \today}}

\maketitle

\begin{abstract}
Cellular networks are usually modeled by placing the base stations on a grid, with mobile users either randomly scattered or placed deterministically.  These models have been used extensively but suffer from being both highly idealized and not very tractable, so complex system-level simulations are used to evaluate coverage/outage probability and rate. More tractable models have long been desirable. We develop new general models for the multi-cell signal-to-interference-plus-noise ratio (SINR) using stochastic geometry.  Under very general assumptions, the resulting expressions for the downlink SINR CCDF (equivalent to the coverage probability) involve quickly computable integrals, and in some practical special cases can be simplified to common integrals (e.g., the Q-function) or even to simple closed-form expressions. We also derive the mean rate, and then the coverage gain (and mean rate loss) from static frequency reuse.  We compare our coverage predictions to the grid model and an actual base station deployment, and observe that the proposed model is pessimistic (a lower bound on coverage) whereas the grid model is optimistic, and that both are about equally accurate. In addition to being more tractable, the proposed model may better capture the increasingly opportunistic and dense placement of base stations in future networks.  \end{abstract}

\section{Introduction}

Cellular systems are now nearly universally deployed and are under ever-increasing pressure to increase the volume of data they can deliver to consumers.  Despite decades of research, tractable models that accurately model other-cell interference (OCI) are still unavailable, which is fairly remarkable given the size of the industry. This deficiency has impeded the development of techniques to combat OCI, which is the most important obstacle to higher spectral efficiency in today's cellular networks, particularly the dense ones in urban areas that are under the most strain.  In this paper we develop accurate and tractable models for downlink capacity and coverage, considering full network interference.

\subsection{Common Approaches and Their Limitations}

A tractable but overly simple downlink model commonly used by information theorists is the Wyner model \cite{Wyn94,ShaWyn97,SomSha00}, which is typically one-dimensional and presumes a unit gain from each base station to the tagged user and an equal gain that is less than one to the two users in the two neighboring cells.  This is a highly inaccurate model unless there is a very large amount of interference averaging over space, such as in the uplink of heavily-loaded CDMA systems \cite{XuZhaAnd10}.  This philosophical approach of distilling other-cell interference to a fixed value has also been advocated for CDMA in \cite{VitVit94} and used in the landmark paper \cite{GilJac91}, where other-cell interference was modeled as a constant factor of the total interference.  For cellular systems using orthogonal multiple access techniques such as in LTE and WiMAX, the Wyner model and related mean-value approaches are particularly inaccurate, since the SINR values over a cell vary dramatically.  Nevertheless, it has been commonly used even up to the present to evaluate the ``capacity'' of multicell systems under various types of multicell cooperation \cite{SomZai07,JinTse08EURASIP,SimSom09}. Another common analysis approach is to consider a single interfering cell \cite{Cha11}.  In the two cell case, at least the SINR varies depending on the user position and possibly fading, but naturally such an approach still neglects most sources of interference in the network and is highly idealized. A recent discussion of such models for the purposes of base station cooperation to reduce interference is given in \cite{GesHan10}. That such simplified approaches to other-cell interference modeling are still considered state-of-the-art for analysis speaks to the difficulty in finding more realistic tractable approaches.

On the other hand, practicing systems engineers and researchers in need of more realistic models typically model a 2-D network of base stations on a regular hexagonal lattice, or slightly more simply, a square lattice, see Fig. \ref{fig:reg}.  Tractable analysis can sometimes be achieved for a \emph{fixed} user with a small number of interfering base stations, for example by considering the ``worst-case'' user location -- the cell corner -- and finding the signal-to-interference-plus-noise ratio (SINR) \cite{RapBook,GolBook}.  The resulting SINR is still a random variable in the case of shadowing and/or fading from which performance metrics like (worst-case) average rate and (worst-case) outage probability relative to some target rate can be determined.  Naturally, such an approach gives very pessimistic results that do not provide much guidance to the performance of most users in the system.  More commonly, Monte Carlo integrations are done by computer, e.g. in the landmark capacity paper \cite{GilJac91}.  Tractable expressions for the SINR are unavailable in general for a random user location in the cell and so more general results that provide guidance into typical SINR or the probability of outage/coverage over the entire cell must be arrived at by complex time-consuming simulations.  In addition to being onerous to construct and run, such private simulations additionally suffer from issues regarding repeatability and transparency, and they seldom inspire ``optimal'' or creative new algorithms or designs.  It is also important to realize that although widely accepted, grid-based models are themselves highly idealized and may be increasingly inaccurate for the heterogeneous and ad hoc deployments common in urban and suburban areas, where cell radii vary considerably due to differences in transmission power, tower height, and user density.  For example, picocells are often inserted into an existing cellular network in the vicinity of high-traffic areas, and short-range femtocells may be scattered in a haphazard manner throughout a centrally planned cellular network.

\subsection{Our Approach and Contributions}

Perhaps counter-intuitively, this paper addresses these long-standing problems by introducing an additional source of randomness: the positions of the base stations.  Instead of assuming they are placed deterministically on a regular grid, we model their location as a homogeneous Poisson point process of density $\lambda$.  Such an approach for BS modelling has been considered as early as 1997 \cite{BacKle97,BacZuy97,Bro00} but the key metrics of coverage (SINR distribution) and rate have not been determined\footnote{The paper \cite{DecMar10} was made public after submission of this paper and contains some similar aspects to the approach in this paper.}. The main advantage of this approach is that the base station positions are all independent which allows substantial tools to be brought to bear from stochastic geometry; see \cite{HaeAnd09} for a recent survey that discusses additional related work, in particular \cite{ChaHan01,ChaAnd09a,YanPet03}.  Although BS's are not independently placed in practice, the results given here can in principle be generalized to point processes that model repulsion or minimum distance, such as determinantal and Matern processes \cite{BacNOW,StoKen96}.  The mobile users are scattered about the plane according to some independent homogeneous point process with a different density, and they communicate with the nearest base station while all other base stations act as interferers, as shown in Fig. \ref{fig:VoronoiCells}.

From such a model, we achieve the following theoretical contributions.  First, we are able to derive a general expression for the probability of coverage in a cellular network where the interference fading/shadowing follows an arbitrary distribution. The coverage probability is the probability that a typical mobile user is able to achieve some threshold SINR, i.e. it is the complementary cumulative distribution function (CCDF).  This expression is not closed-form but also does not require Monte Carlo methods. The coverage is then derived for a number of special cases, namely combinations of (i) exponentially distributed interference power, i.e. Rayleigh fading, (ii) path loss exponent of 4, and (iii) interference-limited networks, i.e. thermal noise is ignored.  These special cases have increasing tractability and in the case that all three simplifications are taken, we derive a remarkably simple formula for coverage probability that depends only on the threshold SINR.  We compare these novel theoretical results with both traditional (and computationally intensive) grid-based simulations and with actual base station locations from a current cellular deployment in a major urban area.  We see that over a broad range of parameter and modeling choices our results provide a reliable lower bound to reality whereas the grid model provides an upper bound that is about equally loose.  In other words, our approach, even in the case of simplifying assumptions (i)-(iii), appears to not only provide simple and tractable predictions but also accurate ones.

Next, we derive the mean achievable rate in our proposed cellular model under similar levels of generality and tractability.  The two competing objectives of coverage and rate are then explored analytically through the consideration of frequency reuse, which is used in some form in nearly all cellular systems\footnote{Even cellular systems such as modern GSM and CDMA networks that claim to deploy universal frequency reuse still thin the interference in time or by using additional frequency bands -- which is mathematically equivalent to thinning in frequency.} to increase the coverage or equivalently the cell edge rates.  Our expressions for coverage and rate are easily modified to include frequency reuse and we find the amount of frequency reuse required to reach a specified coverage probability, as well as seeing how frequency reuse degrades mean rate by using the total bandwidth less efficiently.

\section{Downlink System Model}
\label{sec:model}

The cellular network model consists of base stations (BSs) arranged according to some homogeneous Poisson point process (PPP) $\Phi$ of intensity $\lambda$ in the Euclidean plane. Consider an independent collection of mobile users, located according to some independent stationary point process. We assume each mobile user is associated with the closest base station; namely the users in the Voronoi cell of a BS are associated with it, resulting in coverage areas that comprise a Voronoi tessellation on the plane, as shown in Fig. \ref{fig:VoronoiCells}.  A traditional grid model is shown in Fig. \ref{fig:reg} and an actual base station deployment in Fig. \ref{fig:RealBS}.  The main weakness of the Poisson model is that because of the independence of the PPP, BSs will in some cases be located very close together but with a significant coverage area.  This weakness is balanced by two strengths: the natural inclusion of different cell sizes and shapes and the lack of edge effects, i.e. the network extends indefinitely in all directions.  The models are quantitatively compared in Section \ref{sec:num}.

The standard power loss propagation model is used with path loss exponent $\alpha > 2$.  As far as random channel effects such as fading and shadowing, we assume that the tagged base station and tagged user experience only Rayleigh fading with mean 1, and constant transmit power of $1/\mu$. Then the received power at a typical node a distance $r$ from its base station is $h r^{-\alpha}$ where the random variable $h$ follows an exponential distribution with mean $1/\mu$,  which we denote as $h \sim \exp(\mu)$. Note that other distributions for $h$ can be considered using Prop. 2.2 of \cite{BacBla09} but with some loss of tractability. The interference power follows a general statistical distribution $g$ that could include fading, shadowing, and any other desired random effects.  Simpler expressions result when $g$ is also exponential and these are given as special cases.  Lognormal interference is considered numerically\footnote{ Shadowing is neglected between the tagged BS and user since it can fairly easily be overcome with even slow power control.  In this case the transmit power would be simply $1/g\mu$ and treated as a constant over the shadowing time-scale.}, and we see although it degrades coverage it does not significantly affect the accuracy of our analysis.  Because of these random channel effects, in our model not all users will be connected to the base station capable of providing the highest SINR. All results are for a single transmit and single receive antenna, although future extensions to multiple such antennas are clearly desirable.

The interference power at the typical receiver $I_r$ is the sum of the received powers from all other base stations other than the home base station and is treated as noise in the present work.  There is no same-cell interference, for example due to orthogonal multiple access within a cell.  The noise power is assumed to be additive and constant with value $\sigma^2$ but no specific distribution is assumed in general. The $\snr = \frac{1}{\mu \sigma^2}$ is defined to be the received $\snr$ at a distance of $r = 1$. All analysis is for a typical mobile node which is permissible in a homogeneous PPP by Slivnyak's theorem \cite{StoKen96}.

\section{Coverage}
\label{sec:coverage}

This is the main technical section of the paper, in which we derive the probability of coverage in a downlink cellular network at decreasing levels of generality. The coverage probability is defined as
\begin{equation}
p_c(T,\lambda,\alpha) \triangleq \PP[\sinr > T],
\label{eq:CovDefn}
\end{equation}
and can be thought of equivalently as (i) the probability that a randomly chosen user can achieve a target SINR $T$, (ii) the average fraction of users who at any time achieve SINR $T$, or (iii) the average fraction of the network area that is in ``coverage'' at any time. The probability of coverage is also exactly the CCDF of SINR over the entire network, since the CDF gives $\PP[\sinr \leq T]$.

Without any loss of generality we assume that the mobile user under consideration is located at the origin. A user is in coverage when its SINR from its nearest BS is larger than some threshold $T$  and it is dropped from the network for SINR below $T$.  The SINR of the mobile user at a random distance $r$ from its associated base station can be expressed as
\begin{equation}
\sinr = \frac{h r^{-\alpha}}{\sigma^2 + I_r},
\end{equation}
where
\begin{equation}
I_r = \sum_{i\in\Phi / b_o} g_i R_i^{-\alpha}
\end{equation}
is the cumulative interference from all the other base stations (except the tagged base station for the mobile user at $o$ denoted by $b_o$) which are a distance $R_i$ from the typical user and have fading value $g_i$.

 \subsection{Distance to Nearest Base Station}
\label{subsec:NN}
An important quantity is the distance $r$ separating a typical user from its tagged base station.  Since each user communicates with the closest base station, no other base station can be closer than $r$.  In other words, all interfering base stations must be farther than $r$.  The probability density function (pdf) of $r$ can be derived using the simple fact that the null probability of a 2-D Poisson process in an area $A$ is $\exp(-\lambda A)$.
\begin{eqnarray}
\PP[r > R] & = & \PP[{\rm No ~ BS ~ closer ~ than ~ R}]\\
        &=& e^{-\lambda \pi R^2}.
\end{eqnarray}
Therefore, the cdf is $\PP[r \leq R]  =  F_r(R)  = 1 - e^{-\lambda \pi R^2}$ and the pdf can be found as
\begin{eqnarray}
f_r(r) &=& \frac{\dd F_r(r)}{\dd r} \\
        &=& e^{-\lambda \pi r^2} 2\pi \lambda r.
\end{eqnarray}

Meanwhile, the interference is a standard $M/M$ shot noise \cite{LowTei90,BacNOW,HaenggiNOW} created by a Poisson point process of intensity $\lambda$ outside a disc at center $o$ and of radius $r$, for which some useful results are known and applied in the sequel.

\subsection{General Case and Main Result}

We now state our main and most general result for coverage probability from which all other results in this section follow.

\begin{theorem}
\label{thm:main1}
The probability of coverage of a typical randomly located mobile user in the general cellular network model of Section \ref{sec:model} is
\begin{equation}
p_c(T,\lambda,\alpha) = \pi \lambda \int_0^\infty  e^{-\pi \lambda v\beta(T,\alpha) -\mu T \sigma^2 v^{\alpha/2}} \dd v, \label{eq:Thm1}
\end{equation}
where
\begin{equation}
\beta(T,\alpha)= \frac{2 (\mu T)^{\frac{2}{\alpha}}}{\alpha} \mathbb{E} \left[ g^{\frac{2}{\alpha}} \left(\Gamma(-2/\alpha,\mu T g) -\Gamma(-2/\alpha) \right)\right],
\end{equation}
and the expectation is with respect to the interferer's channel distribution $g$. Also, $\Gamma(a,x) = \int_x^{\infty} t^{a-1}e^{-t}\dd t$  denotes the incomplete gamma function, and $\Gamma(x) = \int_0^{\infty} t^{x-1}e^{-t}\dd t$ the standard gamma function.
\end{theorem}

\begin{proof}Conditioning on the nearest BS being at a distance $r$ from the typical user, the probability of coverage averaged over the plane is
\begin{align*}
p_c(T,\lambda,\alpha) &= \E_r \left[ \PP[\sinr > T\ |\  r] \right] \\
& = \int_{r>0} \PP[\sinr > T\ |\  r ]f_r(r)\dd r \\
& \stackrel{(a)}{=} \int_{r>0} \PP\left[\frac{hr^{-\alpha}}{\sigma^2+I_r} > T\ \Big|\  r\right] e^{-\pi\lambda r^2} 2 \pi \lambda r \dd r \nonumber \\
& =  \int_{r>0} e^{-\pi \lambda r^2} \PP[h>Tr^{\alpha}(\sigma^2+I_r) \ |\  r ] 2 \pi \lambda r \dd r. \nonumber
\end{align*}
The distribution $f_r(r)$ and hence $(a)$ follows from Subsection \ref{subsec:NN}.
Using the fact that $h \sim \exp(\mu)$,   the coverage probability can be expressed as
\begin{eqnarray}
\PP[h>Tr^{\alpha}(\sigma^2+I_r)\ |\  r] &=&  \E_{I_r} \left[\PP[h>Tr^{\alpha}(\sigma^2+I_r) \ |\  r, I_r] \right] \nonumber \\
    &=& \E_{I_r}\left[\exp(-\mu Tr^{\alpha}(\sigma^2+I_r)) \ |\  r \right] \nonumber \\
    &=&   e^{-\mu Tr^{\alpha} \sigma^2}{\cal L}_{I_r}(\mu Tr^{\alpha}),
\end{eqnarray}
where ${\cal L}_{I_r}(s)$ is the Laplace transform of random variable $I_r$ evaluated at $s$ conditioned on the distance  to the closest BS from the origin.  This gives a coverage expression
\begin{equation}
p_c(T,\lambda,\alpha) = \int_{r>0} e^{-\pi \lambda r^2}  e^{-\mu Tr^{\alpha} \sigma^2} {\cal L}_{I_r}(\mu Tr^{\alpha}) 2 \pi \lambda r \dd r. \label{eq:pc1}
\end{equation}
Defining $R_i$ as the distance from the $i$th interfering base station to the tagged receiver and $g_i$ as the interference channel coefficient of arbitrary but identical distribution for all $i$, using the definition of the Laplace transform yields
\begin{align}
{\cal L}_{I_r}(s) &= \E_{I_r}[e^{-sI_r}] = \E_{\Phi, g_i}[\exp(-s\sum_{i\in \Phi\setminus\{ b_o\}} g_i R_i^{-\alpha})] \nonumber\\
&= \E_{\Phi, \{g_i\}}\left[\prod_{i\in \Phi\setminus\{ b_o\}} \exp(-s g_i R_i^{-\alpha})\right] \nonumber\\
&\stackrel{(a)}{=}\E_{\Phi}\left[\prod_{i\in \Phi\setminus \{b_o\}} \E_{g} [\exp(-s g R_i^{-\alpha})] \right] \nonumber\\
&=  \exp\left(-2\pi\lambda \int_r^{\infty} \left(1 - \E_{g}[\exp(-s g v^{-\alpha})]  \right) v  \dd v \right),
\label{eq:laplace}
\end{align}
where $(a)$ follows from the i.i.d. distribution of $g_i$ and its further independence from the point process $\Phi$, and the last step follows from the probability generating functional (PGFL) \cite{StoKen96} of the PPP, which states for some function $f(x)$ that  $\mathbb{E}\left[ \prod_{x\in \Phi} f(x)\right] = \exp\left(-\lambda \int_{\mathbb{R}^2} (1-f(x)) \dd x\right)$.  The integration limits are from $r$ to $\infty$ since the closest interferer is at least at a distance $r$. Let $f(g)$ denote the PDF of $g$. Plugging in $s= \mu Tr^{\alpha}$,  and swapping the integration order gives,
\begin{align*}
{\cal L}_{I_r}(\mu Tr^{\alpha}) &=
\exp\left(-2\pi\lambda\int_0^\infty\left( \int_r^\infty ( 1- e^{-\mu Tr^\alpha v^{-\alpha}g})v \dd v \right)f(g) \dd g  \right).
\end{align*}
 The inside integral can be evaluated by using the change of variables $v^{-\alpha}\rightarrow y$, and the Laplace transform is
\[{\cal L}_{I_r}(\mu Tr^{\alpha})= \exp\left(\lambda\pi r^2- \frac{2\pi\lambda (\mu T)^{\frac{2}{\alpha}}r^2}{\alpha} \int_0^\infty g^{\frac{2}{\alpha}} \left[\Gamma(-2/\alpha,\mu T g) -\Gamma(-2/\alpha)\right] f(g) \dd g  \right).\]
Combining with \eqref{eq:pc1}, and using the substitution $r^2 \rightarrow v$, we obtain the result.
\end{proof}

In short, Theorem 1 gives a general result for the probability of achieving a target SINR $T$ in the network.  It is not closed-form but the integrals are fairly easy to evaluate.  We now turn our attention to a few relevant special cases where significant simplification is possible.

\subsection{Special Cases: Interference Experiences General Fading}

The main simplifications we will now consider in various combinations are (i) allowing the path loss exponent $\alpha = 4$, (ii) an interference-limited network, i.e. $1/\sigma^2 \to \infty$, which we term ``no noise'' and (iii) interference fading power $g \sim \exp(\mu)$ rather than following an arbitrary distribution\footnote{The interference power is also attenuated by the path loss so the \emph{mean} interference power for each base station is less than the \emph{mean} desired power, by definition, even though the fading distributions have the same mean $\mu$, which is a proxy for the transmit power.}.  In this subsection we continue assume the interference power follows a general distribution, so we consider two special cases corresponding to (i) and (ii) above.

\subsubsection{General Fading, Noise, $\alpha = 4$}

First, if $\alpha = 4$, Theorem 1 admits a form that can be evaluated according to
\begin{equation}
\int_0^\infty e^{-a x} e^{-b x^2} \dd x= \sqrt{\frac{\pi}{b}} \exp\left(\frac{a^2}{4b}\right) Q\left(\frac{a}{\sqrt{2 b}}\right),
\label{eq:exp-int}
\end{equation}
where $Q(x) = \frac{1}{\sqrt{2\pi}}\int_x^{\infty} \exp(-y^2/2)\dd y$ is the standard Gaussian tail probability.  Setting $a = \pi \lambda\beta(T,\alpha)$ and $b = \mu T \sigma^2 = T/\snr$ gives
\begin{eqnarray}
p_c(T,\lambda,4) & = & \frac{\pi^{\frac{3}{2}}\lambda}{\sqrt{T/\snr }}  \exp\left(\frac{ (\lambda \pi \beta(T,4))^2 }{4 T/\snr}\right) Q\left(\frac{\lambda \pi\beta(T,4)}{ \sqrt{2 T/\snr }}\right) \label{eq:Special1}.
\end{eqnarray}

Therefore, given the numerical calculation of $\beta(T,4)$ for a chosen interference distribution, the coverage probability can be found in quasi-closed form since $Q(x)$ can be evaluated nearly as easily as a basic trigonometric function by modern calculators and software programs.

\subsubsection{General Fading, No Noise, $\alpha > 2$}
\label{sec:NoNoise}

In most modern cellular networks thermal noise is not an important consideration.  It can be neglected in the cell interior because it is very small compared to the desired signal power (high SNR), and also at the cell edge because the interference power is typically so much larger (high INR).  If $\sigma^2 \to 0$ (or transmit power is increased sufficiently), then using Theorem 1 it is easy to see that
\begin{equation}
p_c(T,\lambda,\alpha) = \frac{1}{\beta(T,\alpha)}.
\label{eq:no-noise}
\end{equation}
In the next subsection, we show that \eqref{eq:Special1} does in fact reduce to \eqref{eq:no-noise} as $\sigma^2 \to 0$, which is not obvious by inspection.

It is interesting to note that in this case the probability of coverage does not depend on the base station density $\lambda$. It follows that both very dense and very sparse networks have a positive probability of coverage when noise is negligible.  Intuitively, this means that increasing the number of base stations does not affect the coverage probability, because the increase in signal power is exactly counter-balanced by the increase in interference power.  This matches empirical observations in interference-limited urban networks as well as predictions of traditional, less-tractable models.  In interference-limited networks, increasing coverage probability typically requires interference management techniques, for example frequency reuse, and not just the deployment of more base stations.  Note that deploying more base stations does allow more users to be simultaneously covered in a given area, both in practice and under our model, because we assume one active user per cell.

\subsubsection{General Fading, Small but Non-zero Noise}
\label{sec:GenNoNoise}
A potentially useful low noise approximation of the success probability can be obtained that is more easily computable than the constant noise power expression and more accurate than the no noise approximation for $\sigma^2 \neq 0$.  Using the expansion $\exp(-x)=1-x+o(x)$, $x\rightarrow 0$ it can be found after an integration by parts of \eqref{eq:Thm1} that
\begin{align}
p_c(T,\lambda,\alpha)&=\frac{1}{\beta(T,\alpha)}- \frac{\mu T \sigma^2(\lambda\pi)^{-\alpha/2}}{\beta(T,\alpha)}\Gamma\left(1+\frac{\alpha}{2}\right)  + o(\sigma^2)
\label{eq:SmallNoise}
\end{align}
For the special case of $\alpha = 4$, it is not immediately obvious that \eqref{eq:Special1} is equivalent to \eqref{eq:no-noise} as $\sigma^2 \to 0$, but indeed it is true.  It is possible to write \eqref{eq:Special1} as
\begin{equation}
p_c(T,\lambda,4) = \frac{\pi^{\frac{3}{2}}\lambda \sqrt{2}}{a} x Q(x)\exp\left(\frac{x^2}{2}\right)
\label{eq:CovSmallNoise}
\end{equation}
where $x = \frac{a}{\sqrt{2b}}$ and $a,b$ as before.  The series expansion of $Q(x)$ for large $x$ is
\begin{equation}
Q(x) = \frac{1}{\sqrt{2\pi}}\exp\left(-\frac{x^2}{2}\right)\left[\frac{1}{x} - \frac{1}{x^2} + o(x^{-4})\right]
\end{equation}
which means that
\begin{equation}
\lim_{x \to \infty} x Q(x)\exp\left(\frac{x^2}{2}\right) = \frac{1}{\sqrt{2 \pi}},
\end{equation}
which allows simplification of \eqref{eq:CovSmallNoise} to \eqref{eq:no-noise} for the case of no noise.

\subsection{Special Cases: Interference is Rayleigh Fading}

Significant simplification is possible when the interference power follows an exponential distribution, i.e. interference experiences Rayleigh fading and shadowing is neglected.  We give the coverage probability for this case as Theorem 2.

\begin{theorem}
\label{thm:main2}
The probability of coverage of a typical randomly located mobile user experiencing exponential interference is
\begin{equation}
p_c(T,\lambda,\alpha) = \pi \lambda \int_0^\infty  e^{-\pi \lambda v (1 + \rho(T,\alpha)) -\mu T \sigma^2 v^{\alpha/2}} \dd v,
\label{eq:expInt}
\end{equation}
where
\begin{equation}
\rho(T,\alpha)= T^{2/\alpha}\int_{T^{-2/\alpha}}^\infty \frac{1}{1+u^{\alpha/2}} \dd u.
\end{equation}
\end{theorem}
\begin{proof}
The proof is a special case of Theorem \ref{thm:main1}, but however lends to much simplification. The proof is provided in Appendix \ref{app:proof1}.
\end{proof}

We now consider the special cases of no noise and $\alpha = 4$.

\subsubsection{Exponential Fading, Noise, $\alpha = 4$}

When $\alpha=4$, using the same approach as in \eqref{eq:exp-int}, we get
\begin{eqnarray}
p_c(T,\lambda,4) & = & \frac{\pi^{\frac{3}{2}}\lambda}{\sqrt{T/\snr }}  \exp\left(\frac{ (\lambda \pi \kappa(T))^2 }{4 T/\snr}\right) Q\left(\frac{\lambda \pi\kappa(T)}{ \sqrt{2 T/\snr }}\right) \label{eq:CovConstNoise},
\end{eqnarray}
where $\kappa(T) = 1+\rho(T,4)= 1 + \sqrt{T}(\pi/2 -\arctan(1/\sqrt{T}))$.

This expression is quite simple and is practically closed-form, requiring only the computation of a simple $Q(x)$ value.
\subsubsection{Exponential Fading, No Noise, $\alpha > 2$}

In the no noise case the result is very similar to general fading in appearance, i.e.
\begin{equation}
p_c(T,\lambda,\alpha) = \frac{1}{1 + \rho(T,\alpha)},
\end{equation}
with $\rho(T,\alpha)$ being faster and easier to compute than the more general expression $\beta(T,\alpha)$.  When the path loss exponent $\alpha=4$, the no noise coverage probability can be further simplified to
\begin{equation}
p_c(T,\lambda,4) =  \frac 1 {1+\sqrt{T} (\pi/2-\arctan(1/\sqrt{T}))} \label{eq:CovNoNoise}.
\end{equation}
This is a remarkably simple expression for coverage probability that depends only on the SIR threshold $T$, and as expected it goes to 1 for $T \to 0$ and to 0 for $T \to \infty$.  For example, if $T = 1$ (0 dB, which would allow a maximum rate of 1 bps/Hz), the probability of coverage in this fully loaded network is $4(4+\pi)^{-1} = 0.56$. This will be compared in more detail to classical models in Section \ref{sec:num}.  A small noise approximation can be performed identically to the procedure of Section \ref{sec:GenNoNoise} with $1+ \rho(T,\alpha)$ replacing $\beta(T,\alpha)$ in \eqref{eq:SmallNoise}.

\section{Average Achievable Rate}
\label{sec:rate}
In this section, we turn our attention to the mean data rate achievable over a cell.  Specifically we compute the mean rate in units of nats/Hz (1 bit $= \ln(2) = 0.693$ nats) for a typical user where adaptive modulation/coding is used so each user can set their rate such that they achieve Shannon bound for their instantaneous SINR, i.e. $\ln(1+\sinr)$.  Interference is treated as noise which means the true channel capacity is not achieved, which would require a multiuser receiver \cite{Cov72,RimUrb96,ElGKimNotes}, but future work could relax this constraint within the random network framework, see e.g. \cite{WebAnd07,BloJin09}. In general, almost any type of modulation, coding, and receiver structure can be easily treated by adding a gap approximation to the rate expression, i.e. $\tau \to \ln(1+\sinr/G)$ where $G \geq 1$ is the gap. The technical tools and organization are similar to Section \ref{sec:coverage} so the discussion will be more concise.  The results are all for exponentially distributed interference power but general distributions could be handled as well following the approach of Theorem \ref{thm:main1} and techniques from \cite{BacBla09}.

\subsection{General Case and Main Result}
We begin by stating the main rate theorem that gives the ergodic capacity of a typical mobile user in the downlink.
\begin{theorem}
\label{thm:rate}
The average ergodic rate of a typical mobile user and its associated base station in the downlink is
\begin{eqnarray}
\tau(\lambda,\alpha) & \triangleq & \E[\ln(1+\sinr)]\\
 & = & \int_{r>0} e^{-\pi \lambda r^2} \int_{t>0} e^{-\sigma^2\mu r^{\alpha}(e^t-1)} {\cal L}_{I_r}(\mu r^{\alpha}(e^t-1)) \dd t  2 \pi \lambda r \dd r,
\end{eqnarray}
where
\begin{equation}
{\cal L}_{I_r}(\mu r^{\alpha}(e^t-1))= \exp\left(-\pi \lambda r^2  (e^{t}-1)^{2/\alpha} \int_{(e^{t}-1)^{-2/\alpha}}^\infty \frac{1}{1+x^{\alpha/2}}  \dd g \right).
\end{equation}
\end{theorem}
\begin{proof}
The proof is provided in Appendix \ref{app:thm3}.
\end{proof}
The computation of $\tau$ in general requires three numerical integrations.

\subsection{Special Case: $\alpha=4$}
For $\alpha=4$ the mean rate simplifies to
\begin{align}
\tau(\lambda,4) & =  \int\limits_{t>0} \int\limits_{r>0} e^{-\sigma^2 \mu r^{4}(e^t-1)}e^{-\pi \lambda r^2 \left(1+\sqrt{e^t-1}
\left( \pi/2-\arctan(1/\sqrt{e^{t}-1})\right)\right)} 2 \pi \lambda r \dd r \dd t.\nonumber \\
& =  \int\limits_{t>0} \int\limits_{v>0} e^{-\sigma^2 \mu v^{2}(e^t-1)/(\pi \lambda)^2} e^{-v \left(1+ \sqrt{e^t-1}  \left(\pi/2-\arctan(1/\sqrt{e^{t}-1})\right)\right)}\dd v \dd t\nonumber
\end{align}
Using \eqref{eq:exp-int},
\begin{align}
\tau(\lambda,4) &= \int\limits_{t>0} \sqrt{\frac{\pi}{b(t)}}\exp\left(\frac{a(t)^2}{4 b(t)}\right)Q\left(\frac{a(t)}{\sqrt{2b(t)}}\right)\dd t \label{eq:rateAlpha},
\end{align}
where $a(t)= 1+ \sqrt{e^t-1}  \left(\pi/2-\arctan(1/\sqrt{e^{t}-1})\right)$ and $b(t) =\sigma^2 \mu (e^t-1)/(\pi \lambda)^2$.  The final expression \eqref{eq:rateAlpha} be evaluated numerically with one numerical integration (presuming an available look up table for $Q(x)$).


\subsection{Special Case: No Noise}
When noise is neglected, $\sigma^2 \to 0$, so
\begin{eqnarray}
\tau(\lambda,\alpha)
& =  & \int_{r>0} \int_{t>0} \exp\left(-\pi \lambda r^2  \left(1+(e^t-1)^{2/\alpha}
\int_{(e^t-1)^{-2/\alpha}}^\infty \frac{1}{1+x^{\alpha/2}}  \dd x \right)\right) 2 \pi \lambda r \dd r \dd t \nonumber\\
& = & \int_{t>0} \int_{r>0} \exp\left(-v \left(1+(e^t-1)^{2/\alpha} \int_{(e^t-1){-2/\alpha}}^\infty \frac{1}{1+x^{\alpha/2}} \dd x \right)\right) \dd v \dd t \nonumber\\
& = & \int_{t>0} \frac 1   {1+(e^t-1)^{2/\alpha} \int_{(e^t-1)^{-2/\alpha}}^\infty \frac{1}{1+x^{\alpha/2}} \dd x} \dd t,
\label{eq:rate_nonoise}
\end{eqnarray}
a quantity that again does not depend on $\lambda$. As in the case of coverage, increasing the base station density does not increase the interference-limited ergodic capacity per user in the downlink because the distance from the mobile user to the nearest base station and the average distance to the nearest interferer both scale like $\Theta(\lambda^{-1/2})$, which cancel.  Note, however, that the overall sum throughput and area spectral efficiency of the network \emph{do increase linearly} with the number of base stations since the number of active users per area achieving rate $\tau$ is exactly $\lambda$, assuming that the user density is sufficiently large such that there is at least one mobile user per cell.

In the particular case of $\alpha=4$ in conjunction with no noise,
\begin{eqnarray*}
(e^t-1)^{2/\alpha} \int_{(e^t-1)^{-2/\alpha}}^\infty \frac{1}{1+x^{\alpha/2}}  \dd x = \sqrt{e^t-1} \left(\pi/2-\arctan(1/\sqrt{e^{t}-1})\right),
\end{eqnarray*}
so the mean rate is expressed to a single simple numerical integration that yields a precise scalar
\begin{eqnarray}
\tau(\lambda,4) = \int_{t>0}\frac 1  {1+  \sqrt{e^t-1}  \left(\pi/2-\arctan(1/\sqrt{e^{t}-1})\right)} \dd t\approx 1.49 \rm{nats/sec/Hz}.
\label{eq:rate_alpha4}
\end{eqnarray}
In other words, our model predicts that the no noise limit for mean downlink rate in a cellular system with Rayleigh fading is 2.15 bps/Hz if $\alpha = 4$.

\section{Validation of the Proposed Model}
\label{sec:num}

Now that we have developed expressions for the coverage and mean rate for cellular networks, it is important to see how these analytical results compare with the widely accepted grid model.  Further, we were able to obtain precise coordinates for base stations over a large urban area from a major service provider, and we compare our results to the coverage predicted by those locations as well (which are neither a perfect grid nor Poisson). Intuitively, we would expect the Poisson model to give pessimistic results compared to a planned deployment due to the strong interference generated by nearby base stations. The grid model is clearly an upper bound since a perfectly regular geometry is in fact optimal from a coverage point of view \cite{Bro00}.  An additional source of optimism in the grid model is the customary neglect of background interference from outer tier base stations.  We see in Section \ref{sec:regular} that the latter effect is not very significant, however.

\subsection{The Grid Model and An Actual BS Deployment}

\label{sec:regular}
A periodic grid is typically used in prior work to model the base station locations. We use a square lattice for notational simplicity but a hexagonal one can also be used all results will only differ by a very small constant.  We consider a home base station located at the origin and $N$ interfering base stations located in square tiers around the home base station. Each tier is a distance $2R$ from the previous tier, i.e. each base station coverage area is a $2R \times 2R$ square, and so any user within a distance $R$ of a base station is guaranteed to be covered by it.  The base station density in this case is $1/4R^2$ base stations per unit area.  A two tier example with $N = 24$ is shown in Figure \ref{fig:reg}.  The SINR for a regular base station deployment becomes
\begin{equation}
\sinr = \frac{h r_u^{-\alpha}}{I_u + \sigma^2},
\end{equation}
where $r_u = \sqrt{x_u^2 + y_u^2}$ with $x_u \sim U[-R,R]$ and $y_u \sim U[-R,R]$.  The channel fading power is still $h \sim \exp(\mu)$ as in previous sections. The interference to the tagged user is now
\begin{equation}
I_u = \sum_{i=1}^N g_i r_i^{-\alpha}
\end{equation}
where $r_i = \sqrt{(x_i-x_u)^2 + (y_i - y_u)^2}$ is the distance seen from interfering base station $i$ and $g_i$ its observed fading power.  The probability of coverage is
\begin{equation}
p_c(T,\alpha) = \PP[\sinr>T] = \PP[h > r_u^{\alpha} T(I_u + \sigma^2)],
\end{equation}
which is no different in principle than \eqref{eq:CovDefn}, but due to the structure of $I_u$ it is difficult to proceed analytically, and so numerical integration is used to compare with the results of Sections \ref{sec:coverage} and \ref{sec:rate}.  One important difference between the behavior of grid and random BS models are the allowed extremes on the distances of the tagged and interfering base stations.  In a grid model, there is always a base station within a specified distance $R$ and never an interfering one closer than $R$.  In the proposed model, two base stations can be arbitrarily close together and hence there is no lower bound on $R$, so both the tagged and an interfering base station can be arbitrarily close to the tagged user. The only constraint is that the interfering base station must be farther than the tagged one, or else a handoff would occur.

We have also obtained the coordinates of a current base station deployment by a major service provider in a relatively flat uniform urban area.  This deployment stretches over an approximately $100 \times 100$ km square, and we show a zoom of the middle $40 \times 40$ km in Fig. \ref{fig:RealBS}.  In this figure\chr{,} the cell boundaries correspond to a Voronoi tessellation and hence are only a function of Euclidean distance, whereas in practice other factors might determine the cell boundaries.  Clearly this is only a single deployment and further validation should be done.  However, we strongly suspect that deployments in many cities follow an even less regular topology due to irregular terrain such as large hills and water features and/or high concentrated population centers.  It seems such scenarios might be even better suited to a random spatial model that the example provided here.

\subsection{Coverage Comparison}

In Fig. \ref{fig:cov1}, we compare the traditional square grid model to the random PPP base station model.  The plot gives the probability that a given $\sinr$ target $T$ on the x-axis can be achieved, i.e. it gives the complementary cumulative distribution function (CCDF) of $\sinr$, i.e. $\PP[\sinr > T]$. Both $N=8$ and $N=24$ are used, and it can be seen that the $N=8$ case is only slightly more optimistic as opposed to $N=24$, at least for $\alpha =4$ (the gap increases slightly for smaller $\alpha$).  The curves all exhibit the same basic shape and as one would expect, a regular grid provides a higher coverage area over all possible $\sinr$ targets.  A small ($< 1$ dB) gap is seen between the $\snr = 10$ and $\snr \to \infty$ cases, which confirms that noise is not a very important consideration in dense cellular networks, which are known to be interference-limited. Therefore, we neglect noise in the ensuing plots.

In Fig. \ref{fig:cov2} we now compare the three different base station location models with exponential (Rayleigh fading) interference.  The random BS model is indeed a lower bound and the grid model an upper bound.  The random BS model appears no worse than the grid model in terms of accuracy and may be preferable from the standpoint that it provides conservative predictions, as well as being much more analytically tractable.  The Poisson BS model becomes more accurate at lower path loss exponents.  There are two reasons for this.  First, the PPP models distant interference whereas a 1 or 2 tier grid model does not; and the the interference of far-off base stations is more significant for small $\alpha$. Second, since a weakness of the Poisson model is the artificially high probability of a nearby and dominant interfering base station, at lower path loss exponents, perhaps counter-intuitively, such an effect is less corrupting because a dominant base station contributes a lower fraction of the total interference due to the slower attenuation of non-dominant interferers.

Next we consider the effect on lognormal interference, which is common in cellular networks.  Whereas shadowing to the desired base station can be overcome with power control (or macrodiversity, not considered here) the interference remains lognormal.  We assume the shadowing is given by a value $10^{\frac{X}{10}}$ where $X \sim N(\xi,\kappa^2)$ and $\xi$ and $\kappa$ are now in dB.  We normalize $\xi$ to be the same as for the exponential case and consider various values of $\kappa$ in Figs. \ref{fig:LN1} and \ref{fig:LN2}. Fig. \ref{fig:LN1} shows the extent to which lognormal interference increases the coverage probability in our model, whereas Fig. \ref{fig:LN2} shows that our model still reasonably tracks a real deployment. It may seem counterintuitive that increasing lognormal interference increases the coverage probability, the reason being that cell edge users have poor mean SINR (often below $T$), and so increasing randomness gives them an increasing chance of being in coverage.  It also implies that SINR-aware scheduling, which is not considered here, might be able to significantly increase coverage.


\section{Frequency Reuse: Coverage vs. Rate}

Cellular network operators must provide at least some coverage to their customers with very high probability.  For example,  $\sinr = 1$ might be a minimal level of quality needed to provide a voice call.  In this case, for $\alpha = 4$ we can see from Fig. \ref{fig:cov1} that the grid model gives a success probability of about 0.7 and the PPP model predicts 0.53.  Clearly, neither is sufficient for a commercial network, so cellular designers must find a way to increase the coverage probability.  Assuming the network is indeed interference-limited, a common way to do this is to reduce the number of interfering base stations.  This can be done statically through a planned and fixed frequency reuse pattern and/or cell sectoring, or more adaptively via a reduced duty cycle in time (as in GSM or CDMA voice traffic), fractional frequency reuse, dynamic bandwidth allocation, or other related approaches \cite{LTEBook,Boudreau2009}.  More sophisticated interference cancellation/suppression approaches can also be used, potentially utilizing multiple antennas.  In this paper,  we restrict our attention to  straightforward per-cell frequency reuse.

In frequency reuse, the reuse factor $\delta \geq 1$ determines the number of different frequency bands used by the network, where just one band is used per cell.  For example, if $\delta = 2$ then the square grid of Fig. \ref{fig:reg} can assign the top row base stations frequencies 1, 2, 1, 2, 1, and then the second row 2, 1, 2, 1, 2, and so on.  In this way interfering base stations are now separated by a distance $2\sqrt{2}R $ rather than $2R$. Larger values of $\delta$ monotonically decrease the interference, e.g. $\delta = 4$ allows a base station separation of $4R$ in the square grid model.

The PPP BS model also allows for interference thinning, but instead of a fixed pattern (which is not possible in a random deployment) we assume that each base station picks one of $\delta$ bands at random. A visual example is given in Fig. \ref{fig:freqReuse} for $\delta = 4$. The model suffers from the fact that adjacent base stations may simultaneously use the same frequency even for large $\delta$.  However, it is not clear this is any worse of a model than the idealistic grid model with pre-planned frequency reuse, especially for 4G OFDMA-based networks that will use dynamic frequency allocation with very aggressive overall frequency reuse, so that each subcarrier may appear from a birds-eye view of the network to be allocated almost randomly at any snapshot in time.


\subsection{Increasing Coverage via Frequency Reuse}

First, we consider the effect of random frequency reuse on the coverage probability.
\begin{theorem}
\label{thm:freq_outage}
If  $\delta$ frequency bands are randomly allocated to the cells, then the coverage probability with exponentially distributed interference power is equal to
\begin{equation}
p_c(T,\lambda,\alpha,\delta) = \pi \lambda \int_0^\infty  \exp\left(-\pi \lambda v \left(1 + \frac{1}{\delta}\rho(T,\alpha)\right) -\mu T \sigma^2 v^{\alpha/2}\right) \dd v.
\end{equation}
\end{theorem}
\begin{proof}
A typical mobile user at the origin $o$ would be served by its closest BS from the complete point process $\Phi$.  Call this distance $r$, then as in Section \ref{subsec:NN} $r$ is Rayleigh distributed with PDF $f_r(r)=\lambda 2\pi r\exp(-\lambda \pi r^2)$ with the interferers
located outside $r$. The interfering BSs which transmit in the same frequency band are a thinned version of the original PPP and have a density $\lambda/\delta$.  Since a thinned version of a PPP is again a PPP, the rest of the proof exactly follows Theorem \ref{thm:main2}.
\end{proof}

Observe that as the number of frequency bands $\delta \rightarrow \infty$ a coverage limit is reached that depends only on the noise power. In Fig. \ref{fig:coverage_freq}, the coverage probability  is plotted with respect to the $\sinr$ threshold $T$ for $\delta = 2$ and $\delta = 4$ for each of the three BS placement models.  

A cellular operator often wishes to \emph{guarantee} a certain probability of coverage to its customers.  For a given blocking/outage probability $\epsilon$, the following corollary for the no noise case provides the number of frequency bands that are required.
\begin{corollary}
The minimum number of frequency bands needed to ensure an outage probability no greater than $\epsilon$ is
\begin{equation}
\delta = \left\lceil \frac{\rho(T,\alpha)(1-\epsilon)}{\epsilon} \right\rceil.
\end{equation}
\end{corollary}
\begin{proof}
For the case of no noise, $\sigma^2=0$ and hence ${\cal L}_W(\mu Tr^{\alpha})=1$ which simplifies the coverage probability to
\begin{equation}
p_c(T,\lambda,\alpha,\delta)=\frac{1}{1+ \frac{\rho(T,\alpha)}{\delta}},
\label{eq:no_noise_freq}
\end{equation}
from which the result follows by setting this quantity to be equal to $1-\epsilon$ and requiring it to be an integer.
\end{proof}

We now compare these analytical results with central frequency planning (optimal) for the grid model and a heuristic approach for the actual base stations.  Frequency planning in actual cellular networks is a complex optimization problem that depends on the specific geography and data loads and is often performed heuristically by the service provider. For nominally reuse 1 systems like LTE, frequency reuse can also be done adaptively on a per subcarrier basis but here we consider only static frequency reuse.  To provide a reasonably fair comparison, we have used a simple centralized greedy frequency allocation algorithm that maximizes the distance between cells sharing the same frequency band in the actual BS network. Although not ``optimal'', it provides a reasonable benchmark to compare against the performance predictions of PPP model which uses a random frequency allocation. Examples of both allocations are shown visually in Fig. \ref{fig:freqReuse}. Under a random frequency allocation, adjacent cells may transmit in the same band with a higher probability as compared to a planned allocation.  This leads to an increasing gap between the coverage predicted by the PPP model and that of the actual BSs as $\delta$  increases, as per Fig. \ref{fig:coverage_freq}. As expected, central frequency planning also outperforms random frequency allocation.  This indicates that the proposed model in its current form is not a faithful predictor of coverage for large frequency reuse factors.  However, a recent work \cite{NovGan11}uses this approach to study related and more general interference management techniques such as fractional frequency reuse.

\subsection{Frequency Reuse's Effect on Rate}

The desirable increase in coverage with increasing $\delta$ has to be balanced against the fact that each cell can only use $1/\delta$th of the available frequencies.  In this section we will show that the optimal $\delta$ from a mean rate point of view is in fact $\delta = 1$, i.e. any increase in coverage from frequency reuse is paid for by decreasing the overall sum rate in the network.  The following general result can be given for average rate with frequency reuse, the key observation being that since the bandwidth per cell was previously normalized to 1 Hz, now it is $\frac{1}{\delta}$ Hz.  We assume the SNR per band is unchanged.
\begin{theorem}
If  $\delta$ frequency bands are randomly allocated to the cells, the average rate of a typical mobile user in a downlink is
\begin{equation}
\tau(\lambda,\alpha,\delta) = \frac{1}{\delta} \int_{r>0} e^{-\pi \lambda r^2} \int_{t>0} e^{-\sigma^2\mu r^{\alpha}(e^t-1)} {\cal L}_{I_r}(\mu r^{\alpha}(e^t-1)) \dd t  2 \pi \lambda r \dd r,
\label{eq:333}
\end{equation}
where
\begin{equation}{\cal L}_{I_r}(\mu r^{\alpha}(e^t-1))= \exp\left(- \frac{\lambda\pi r^2  (e^{t}-1)^{2/\alpha}}{\delta}  \int_{(e^{t}-1)^{-2/\alpha}}^\infty \frac{1}{1+g^{\alpha/2}}  \dd g \right).\end{equation}
\end{theorem}
\begin{proof}
The average rate of a typical mobile user is $\frac{1}{\delta}\mathbb{E}[\ln(1+\sinr)]$, and the proof proceeds in a similar manner to Theorem \ref{thm:freq_outage} and Theorem \ref{thm:rate} and so is omitted.
\end{proof}
From the above Theorem, the average rate \emph{without noise} is given by
\[
\tau(\lambda,\alpha,\delta) =  \int_{t>0} \frac 1   {\delta+(e^t-1)^{2/\alpha} \int_{(e^t-1)^{-2/\alpha}}^\infty \frac{1}{1+x^{\alpha/2}} \dd x} \dd t,
\]
and is obviously maximized for $\delta=1$. This can also be seen visually in Fig. \ref{fig:rate_freq}, which shows the average rate as a function of $\delta$ for two different path loss exponents. The next corollary generalizes this observation to the case of non-zero noise.
\begin{corollary}
The average rate of a typical mobile user $ \tau(\lambda,\alpha,\delta)$ is maximized for $\delta=1$.
\end{corollary}
\begin{proof}
Using the substitution $r^2 \rightarrow \delta y$ in \eqref{eq:333}, we observe that the integrand decreases with $\delta$, hence verifying the claim.
\end{proof}
As in \eqref{eq:rate_alpha4}, the average rate simplifies further for the case of $\alpha=4$ and no noise and
\[\tau(\lambda,4,\delta)= \int_{t>0}\frac 1  {\delta+  \sqrt{e^t-1}  \left(\pi/2-\arctan(1/\sqrt{e^{t}-1})\right)} \dd t.\]
By using numerical integration, $\tau(\lambda,4,1)\approx 1.49$ nats/sec/Hz, $\tau(\lambda,4,2)\approx 1.1$ {nats/sec/Hz}, and $\tau(\lambda,4,3)\approx 0.87 $ nats/sec/Hz.

\section{Conclusions}

This paper has presented a new framework for downlink cellular network analysis.  It is significantly more tractable than the traditional grid-based models, and appears to track (and lower bound) a real deployment about as accurately as the traditional grid model (which upper bounds).  A final verdict on its accuracy will require extensive comparison with further real base station deployments. In view of current trends whereby base stations are deployed somewhat opportunistically with ever-increasing density, and having variable cell radii, the proposed model may actually become increasingly accurate as well as much more tractable.

Given the number of problems of contemporary interest that require modeling neighboring base stations, the possibilities for future work using this model are extensive. An extension to the uplink would be desirable. Further extensions to this approach could include random spatial placements of base stations that model repulsion, or heterogeneous networks that have both macro and micro/pico/femto cells with differing transmit powers and coverage areas. It would also be of interest to explore how various multiple antenna techniques, opportunistic scheduling, and base station cooperation affect coverage and rate.

\section{Acknowledgements}

The authors appreciate feedback from A. Lozano, N. Jindal, and J. Foschini.  In particular they suggested the consideration of lognormal shadowing, which led to the general fading results.  The detailed feedback of the reviewers was also helpful.

\bibliographystyle{IEEEtran}
\bibliography{Andrews}

\begin{thebibliography}{10}
\providecommand{\url}[1]{#1}
\csname url@samestyle\endcsname
\providecommand{\newblock}{\relax}
\providecommand{\bibinfo}[2]{#2}
\providecommand{\BIBentrySTDinterwordspacing}{\spaceskip=0pt\relax}
\providecommand{\BIBentryALTinterwordstretchfactor}{4}
\providecommand{\BIBentryALTinterwordspacing}{\spaceskip=\fontdimen2\font plus
\BIBentryALTinterwordstretchfactor\fontdimen3\font minus
  \fontdimen4\font\relax}
\providecommand{\BIBforeignlanguage}[2]{{%
\expandafter\ifx\csname l@#1\endcsname\relax
\typeout{** WARNING: IEEEtran.bst: No hyphenation pattern has been}%
\typeout{** loaded for the language `#1'. Using the pattern for}%
\typeout{** the default language instead.}%
\else
\language=\csname l@#1\endcsname
\fi
#2}}
\providecommand{\BIBdecl}{\relax}
\BIBdecl

\bibitem{Wyn94}
A.~D. Wyner, ``Shannon-theoretic approach to a {G}aussian cellular
  multiple-access channel,'' \emph{IEEE Trans. on Info. Theory}, vol.~40,
  no.~11, pp. 1713--1727, Nov. 1994.

\bibitem{ShaWyn97}
S.~{Shamai } and A.~D. Wyner, ``Information-theoretic considerations for
  symmetric, cellular, multiple-access fading channels - parts {I} \& {II},,''
  \emph{IEEE Trans. on Info. Theory}, vol.~43, no.~11, pp. 1877--1911, Nov.
  1997.

\bibitem{SomSha00}
O.~Somekh and S.~Shamai, ``Shannon-theoretic approach to a {Gaussian} cellular
  multi-access channel with fading,'' \emph{IEEE Trans. on Info. Theory},
  vol.~46, pp. 1401--1425, Jul. 2000.

\bibitem{XuZhaAnd10}
J.~Xu, J.~Zhang, and J.~G. Andrews, ``When does the {Wyner} model accurately
  describe an uplink cellular network?'' in \emph{IEEE Globecom}, Miami, FL,
  Dec. 2010.

\bibitem{VitVit94}
A.~J. Viterbi, A.~M. Viterbi, and E.~Zehavi, ``Other-cell interference in
  cellular power-controlled {CDMA},'' \emph{IEEE Trans. on Communications},
  vol.~42, no. 2/3/4, pp. 1501--4, Feb-Apr 1994.

\bibitem{GilJac91}
K.~S. Gilhousen, I.~Jacobs, R.~Padovani, A.~J. Viterbi, L.~Weaver, and
  C.~Wheatley, ``On the capacity of a cellular {CDMA} system,'' \emph{IEEE
  Trans. on Veh. Technology}, vol.~40, no.~2, pp. 303--12, May 1991.

\bibitem{SomZai07}
O.~Somekh, B.~M. Zaidel, and S.~Shamai, ``Sum rate characterization of joint
  multiple cell-site processing,'' \emph{IEEE Trans. on Info. Theory}, pp.
  4473--4497, Dec. 2007.

\bibitem{JinTse08EURASIP}
S.~Jing, D.~N.~C. Tse, J.~Hou, J.~B. Soriaga, J.~E. Smee, and R.~Padovani,
  ``Multi-cell downlink capacity with coordinated processing,'' \emph{EURASIP
  Journal on Wireless Communications and Networking}, 2008, volume 2008,
  Article ID 586878.

\bibitem{SimSom09}
O.~Simeone, O.~Somekh, H.~V. Poor, and S.~Shamai, ``Local base station
  cooperation via finite-capacity links for the uplink of linear cellular
  networks,'' \emph{IEEE Trans. Info. Theory}, vol.~55, no.~1, pp. 190--204,
  Jan. 2009.

\bibitem{Cha11}
C.~B. Chae, I.~Hwang, R.~W. Heath, and V.~Tarokh, ``Interference
  aware-coordinated beamform system in a two-cell environment,'' \emph{IEEE
  Journal on Sel. Areas in Communications}, to appear.

\bibitem{GesHan10}
D.~Gesbert, S.~Hanly, H.~Huang, S.~Shamai, O.~Simeone, and W.~Yu, ``Multi-cell
  mimo cooperative networks: A new look at interference,'' \emph{IEEE Journal
  on Sel. Areas in Communications}, vol.~28, no.~9, pp. 1380 --1408, Dec. 2010.

\bibitem{RapBook}
T.~S. Rappaport, \emph{Wireless Communications: Principles and Practice},
  2nd~ed.\hskip 1em plus 0.5em minus 0.4em\relax Upper Saddle River, New
  Jersey: Prentice-Hall, 2002.

\bibitem{GolBook}
A.~J. Goldsmith, \emph{Wireless Communications}.\hskip 1em plus 0.5em minus
  0.4em\relax Cambridge University Press, 2005.

\bibitem{BacKle97}
F.~Baccelli, M.~Klein, M.~Lebourges, and S.~Zuyev, ``Stochastic geometry and
  architecture of communication networks,'' \emph{J. Telecommunication
  Systems}, vol.~7, no.~1, pp. 209--227, 1997.

\bibitem{BacZuy97}
F.~Baccelli and S.~Zuyev, ``Stochastic geometry models of mobile communication
  networks,'' in \emph{Frontiers in queueing}.\hskip 1em plus 0.5em minus
  0.4em\relax Boca Raton, FL: CRC Press, 1997, pp. 227--243.

\bibitem{Bro00}
T.~X. Brown, ``Cellular performance bounds via shotgun cellular systems,''
  \emph{IEEE Journal on Sel. Areas in Communications}, vol.~18, no.~11, pp.
  2443 -- 55, Nov. 2000.

\bibitem{DecMar10}
L.~Decreusefond, P.~Martins, and T.-T. Vu, ``An analytical model for evaluating
  outage and handover probability of cellular wireless networks,''
  \emph{arXiv:1009.0193v1}, Sep. 2010.

\bibitem{HaeAnd09}
M.~Haenggi, J.~G. Andrews, F.~Baccelli, O.~Dousse, and M.~Franceschetti,
  ``Stochastic geometry and random graphs for the analysis and design of
  wireless networks,'' \emph{IEEE Journal on Sel. Areas in Communications},
  vol.~27, no.~7, pp. 1029--46, Sep. 2009.

\bibitem{ChaHan01}
C.~C. Chan and S.~V. Hanly, ``Calculating the outage probability in a {CDMA}
  network with spatial {Poisson} traffic,'' \emph{{IEEE} Transactions on
  Vehicular Technology}, vol.~50, no.~1, pp. 183--204, Jan. 2001.

\bibitem{ChaAnd09a}
V.~Chandrasekhar and J.~G. Andrews, ``Uplink capacity and interference
  avoidance for two-tier femtocell networks,'' \emph{IEEE Transactions on
  Wireless Communications}, vol.~8, no.~7, pp. 3498--3509, July 2009.

\bibitem{YanPet03}
X.~Yang and A.~Petropulu, ``Co-channel interference modelling and analysis in a
  {Poisson} field of interferers in wireless communications,'' \emph{IEEE
  Trans. on Signal Processing}, vol.~51, no.~1, pp. 64--76, Jan. 2003.

\bibitem{BacNOW}
F.~Baccelli and B.~Blaszczyszyn, \emph{Stochastic Geometry and Wireless
  Networks}.\hskip 1em plus 0.5em minus 0.4em\relax NOW: Foundations and Trends
  in Networking, 2010.

\bibitem{StoKen96}
D.~Stoyan, W.~Kendall, and J.~Mecke, \emph{Stochastic Geometry and Its
  Applications, 2nd Edition}, 2nd~ed.\hskip 1em plus 0.5em minus 0.4em\relax
  John Wiley and Sons, 1996.

\bibitem{BacBla09}
F.~Baccelli, B.~Blaszczyszyn, and P.~Muhlethaler, ``Stochastic analysis of
  spatial and opportunistic aloha,'' \emph{IEEE Journal on Sel. Areas in
  Communications}, pp. 1105--1119, Sept. 2009.

\bibitem{LowTei90}
S.~B. Lowen and M.~C. Teich, ``Power-law shot noise,'' \emph{{IEEE}
  Transactions on Information Theory}, vol.~36, no.~6, pp. 1302--1318, November
  1990.

\bibitem{HaenggiNOW}
M.~Haenggi and R.~K. Ganti, \emph{Interference in Large Wireless
  Networks}.\hskip 1em plus 0.5em minus 0.4em\relax NOW: Foundations and Trends
  in Networking, 2009.

\bibitem{Cov72}
T.~Cover, ``Broadcast channels,'' \emph{IEEE Trans. on Info. Theory}, vol.~18,
  no.~1, pp. 2 -- 14, Jan. 1972.

\bibitem{RimUrb96}
B.~Rimoldi and R.~Urbanke, ``A rate splitting approach to the {Gaussian}
  multiple access channel,'' \emph{IEEE Trans. on Info. Theory}, vol.~42, pp.
  364--75, Mar. 1996.

\bibitem{ElGKimNotes}
A.~E. Gamal and Y.~H. Kim, \emph{Lecture Notes on Network Information
  Theory}.\hskip 1em plus 0.5em minus 0.4em\relax arXiv:1001.3404v4, 2010.

\bibitem{WebAnd07}
S.~Weber, J.~G. Andrews, X.~Yang, and G.~de~Veciana, ``Transmission capacity of
  wireless ad hoc networks with successive interference cancellation,''
  \emph{IEEE Trans. on Info. Theory}, vol.~53, no.~8, pp. 2799--2814, Aug.
  2007.

\bibitem{BloJin09}
J.~Blomer and N.~Jindal, ``Transmission capacity of wireless ad hoc networks:
  Successive interference cancellation vs. joint detection,'' in \emph{Proc.,
  IEEE Intl. Conf. on Communications}, Dresden, Germany, Jun. 2009, pp. 1 -- 5.

\bibitem{LTEBook}
A.~Ghosh, J.~Zhang, J.~G. Andrews, and R.~Muhamed, \emph{Fundamentals of
  LTE}.\hskip 1em plus 0.5em minus 0.4em\relax Prentice-Hall, 2010.

\bibitem{Boudreau2009}
G.~Boudreau, J.~Panicker, N.~Guo, R.~Chang, N.~Wang, and S.~Vrzic,
  ``Interference coordination and cancellation for 4{G} networks,'' \emph{IEEE
  Communications Magazine}, vol.~47, no.~4, pp. 74--81, Apr. 2009.

\bibitem{NovGan11}
T.~D. Novlan, R.~K. Ganti, A.~Ghosh, and J.~G. Andrews, ``Analytical evaluation
  of fractional frequency reuse for {OFDMA} cellular networks,'' \emph{ArXiv
  e-prints}, Jan. 2011, available at http://arxiv.org/abs/1101.5130.

\end{thebibliography}

\appendices
\section{Proof of Theorem \ref{thm:main2}}
\label{app:proof1}
The proof tracks the proof of Theorem 1 up until step (a) of \eqref{eq:laplace}.  Then,
\begin{align}
{\cal L}_{I_r}(s) &= \E_{\Phi,\{g_i\}}\left[\prod_{i\in \Phi\setminus \{b_o\}} \E_{g_i} [\exp(-s g_i R_i^{-\alpha})] \right] \nonumber\\
&= \E_{\Phi}\left[\prod_{i\in \Phi\setminus\{ b_o\}} \frac{\mu}{\mu+s R_i^{-\alpha}}\right] \nonumber\\
&=  \exp\left(-2\pi\lambda \int_r^{\infty} \left(1 - \frac{\mu}{\mu+s v^{-\alpha}}\right) v  \dd v \right),
\end{align}
which admits a much simpler form than \eqref{eq:laplace} due to the new assumption that $g_i \sim \exp(\mu)$.  The integration limits are still from $r$ to $\infty$ and plugging in $s= \mu Tr^{\alpha}$ now gives
\begin{equation*}
{\cal L}_{I_r}(\mu Tr^{\alpha}) = \exp\left(-2\pi \lambda \int_{r}^\infty \frac{T}{T+(v/r)^\alpha} v \dd v\right).
\end{equation*}
Employing a change of variables $u = \left(\frac{v}{rT^{\frac{1}{\alpha}}} \right)^2$ results in
\begin{equation}
{\cal L}_{I_r}(\mu Tr^{\alpha}) = \exp\left(-\pi r^2 \lambda \rho(T,\alpha) \right),
\label{eq:Ir}
\end{equation}
where
\begin{equation*}
\rho(T,\alpha)= T^{2/\alpha}\int_{T^{-2/\alpha}}^\infty \frac{1}{1+u^{\alpha/2}} \dd u.
\end{equation*}
Plugging \eqref{eq:Ir} into \eqref{eq:pc1} with $v \to r^2$ gives the desired result.

\section{Proof of Theorem \ref{thm:rate}}
\label{app:thm3}
The ergodic rate of the typical user is $\tau \triangleq \E[\ln(1+\sinr)]$ where the average is taken over both the spatial PPP and the fading distribution. Since for a  positive random variable $X$, $\E[X]= \int_{t>0} \PP(X>t)\dd t $, it follows similar to Theorems \ref{thm:main1} and \ref{thm:main2} that
\begin{align*}
\tau(\lambda,\alpha) \triangleq \E[\ln(1+\sinr)]
& = \int_{r>0} e^{-\pi \lambda r^2} \E\left(\ln\left(1+\frac{hr^{-\alpha}}{\sigma^2+I_r}\right)\right) 2 \pi \lambda r \dd r\\
& = \int_{r>0} e^{-\pi \lambda r^2} \int_{t>0} \PP\left[\ln\left(1+\frac{hr^{-\alpha}}{\sigma^2+I_r}\right)>t\right] \dd t  2 \pi \lambda r \dd r\\
& = \int_{r>0} e^{-\pi \lambda r^2} \int_{t>0} \PP\left[h>r^{\alpha}(\sigma^2+I_r)(e^t-1)\right] \dd t  2 \pi \lambda r \dd r\\
& = \int_{r>0} e^{-\pi \lambda r^2} \int_{t>0} \E\left( \exp\left(-\mu r^{\alpha}(\sigma^2+I_r)(e^t-1)\right)\right) \dd t  2 \pi \lambda r \dd r\\
& = \int_{r>0} e^{-\pi \lambda r^2} \int_{t>0} e^{-\sigma^2\mu r^{\alpha}(e^t-1)} {\cal L}_{I_r}(\mu r^{\alpha}(e^t-1)) \dd t  2 \pi \lambda r \dd r.
\end{align*}
From \eqref{eq:laplace} we obtain
\begin{align*}
{\cal L}_{I_r}(\mu r^{\alpha}(e^t-1))
& = \exp\left(-2\pi \lambda \int_{r}^\infty \left( 1-\frac{1}{1+(e^t-1)(r/v)^\alpha}\right) v \dd v\right)\\
& =\exp\left(-\pi \lambda r^2 \int_{1}^\infty \frac{e^{t}-1}{e^{t}-1+u^{\alpha/2}} \dd u\right)\\
& =\exp\left(-\pi \lambda r^2  (e^{t}-1)^{2/\alpha} \int_{(e^{t}-1)^{-2/\alpha}}^\infty \frac{1}{1+x^{\alpha/2}}  \dd x \right),
\end{align*}
and the proof is complete.
\newpage

\begin{figure}
\centering
\includegraphics[width=3.5in]{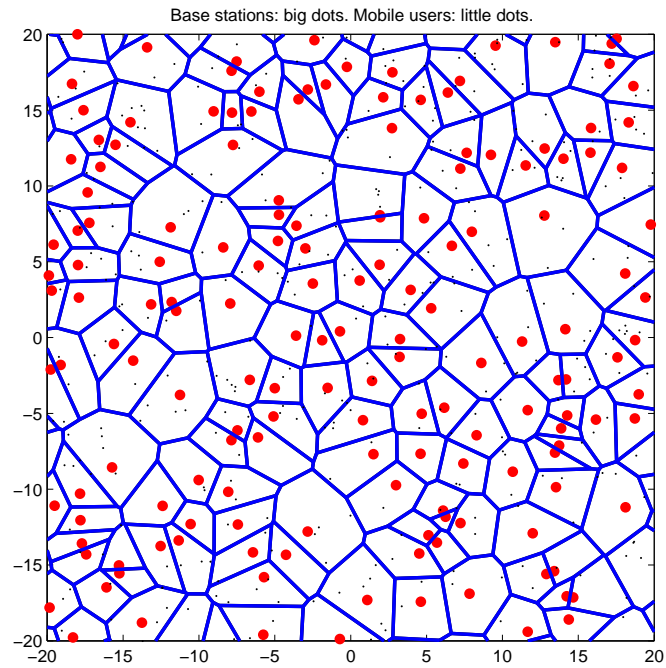}
\caption{Poisson distributed base stations and mobiles, with each mobile associated with the nearest BS. The cell boundaries are shown and form a Voronoi tessellation.}
\label{fig:VoronoiCells}
\end{figure}

\begin{figure}
\centering
\includegraphics[width=3.5in]{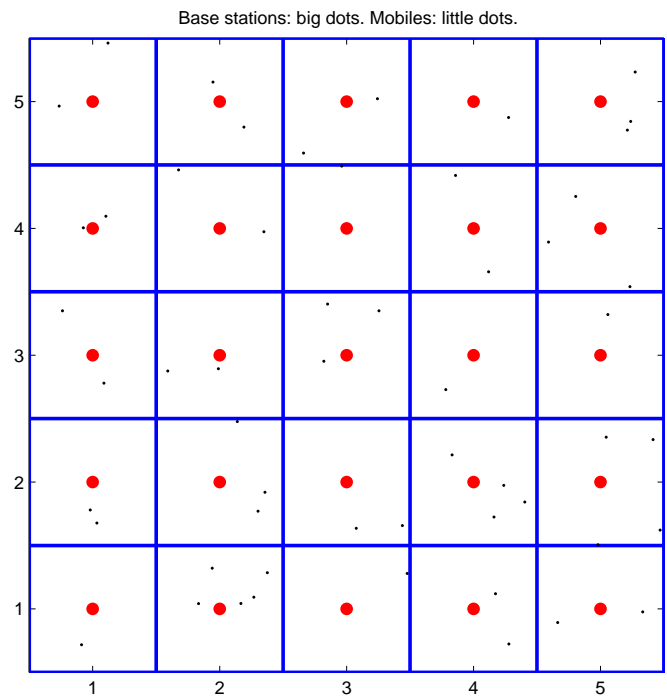}
\caption{A regular square lattice model for cellular base stations with one tier of eight interfering base stations.  The base stations are marked by circles and the active mobile user in the tagged cell by a cross.}
\label{fig:reg}
\end{figure}

\begin{figure}
\centering
\includegraphics[width=3.5in]{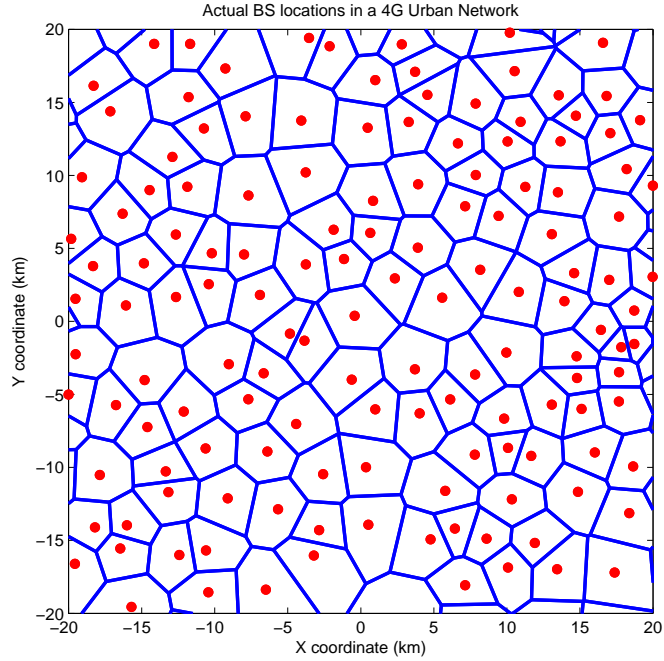}
\caption{A $40 \times 40$ km view of a current base station deployment by a major service provider in a relatively flat urban area, with cell boundaries corresponding to a Voronoi tessellation.}
\label{fig:RealBS}
\end{figure}

\begin{figure}
\centering
\includegraphics[width=4.4in]{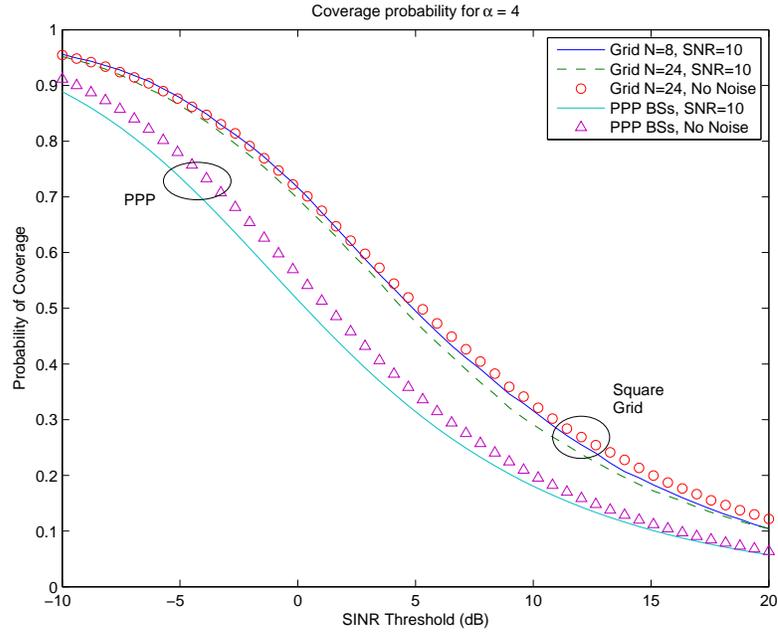}
\caption{Probability of coverage comparison between proposed PPP base station model and square grid model with $N=8,24$ and $\alpha =4$.  The no noise
approximation is quite accurate, and it can be seen there is only a slightly lower coverage area with 24 interfering base stations versus 8.}
\label{fig:cov1}
\end{figure}

\begin{figure}
\centering
\includegraphics[width=3.2in]{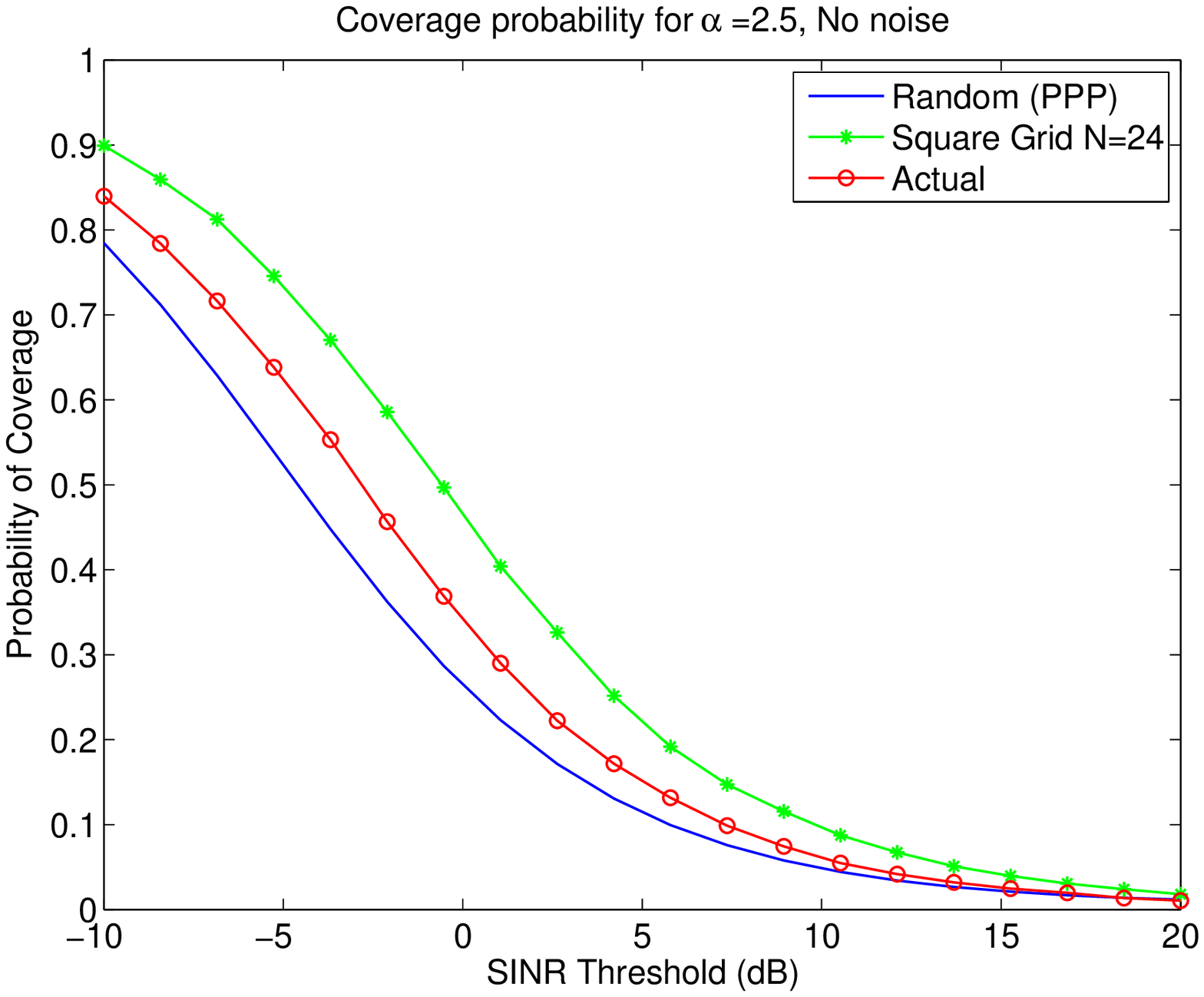}
\includegraphics[width=3.2in]{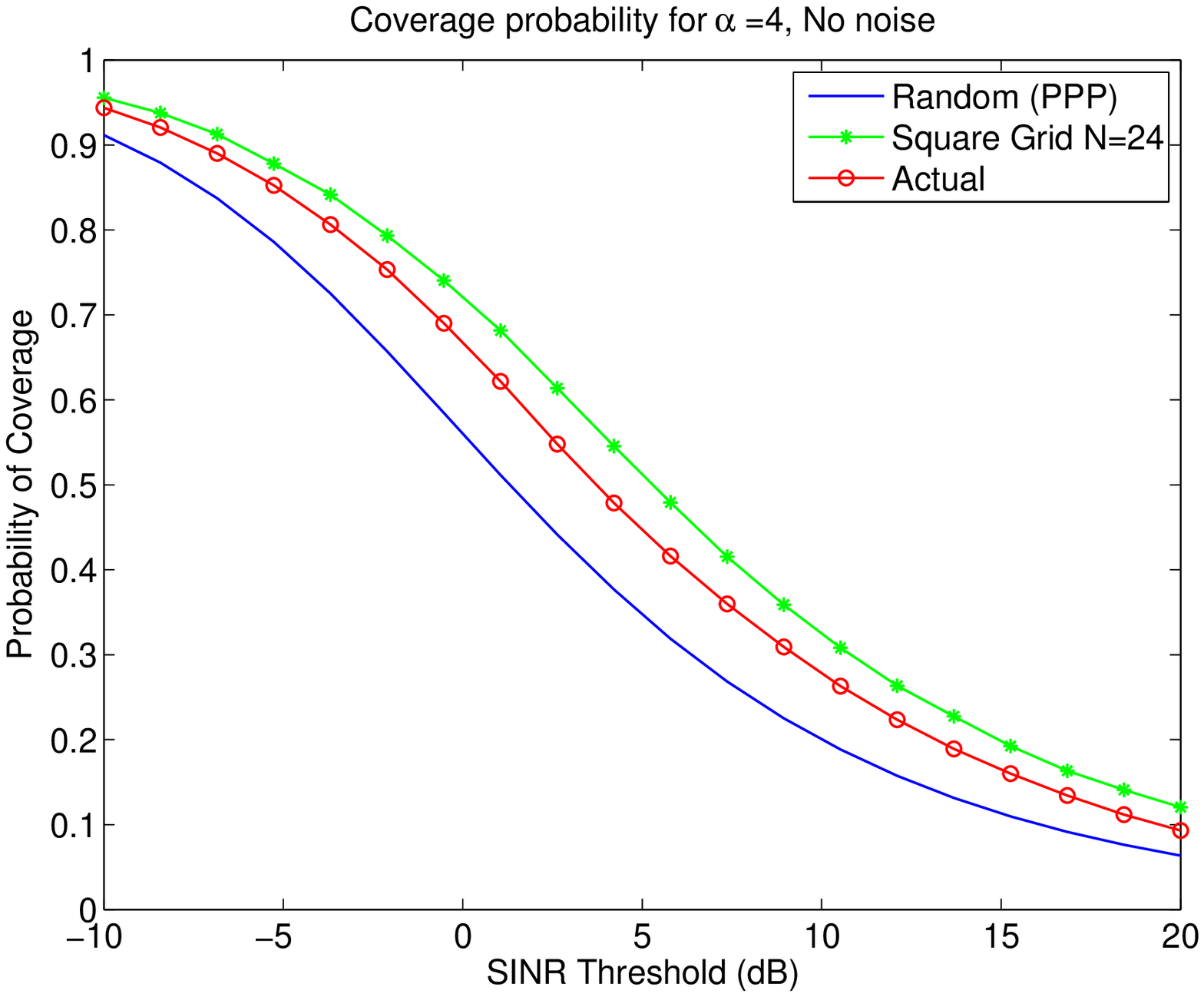}
\caption{Probability of coverage for $\alpha = 2.5$ (left) and $\alpha = 4$ (right), $\snr = 10$, exponential interference. The proposed model is a lower bound and more accurate at lower path loss exponents.}
\label{fig:cov2}
\end{figure}

\begin{figure}
\centering
\includegraphics[width=4.4in]{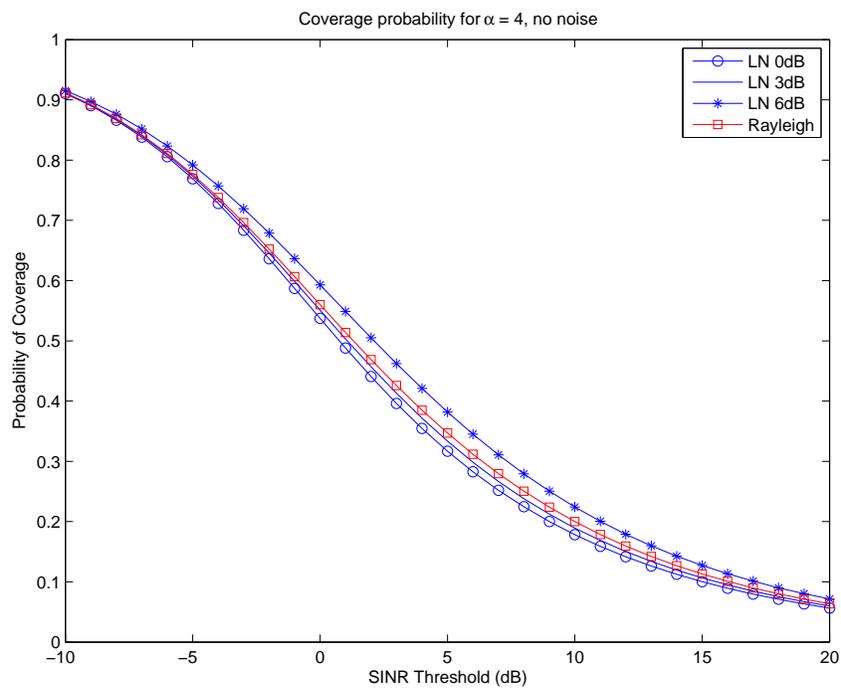}
\caption{Poisson distributed base stations, no noise, $\alpha = 4$ with 4 curves corresponding to lognormal shadowing standard deviations of 0, 3, and 6 dB and Rayleigh fading interference (without shadowing).}
\label{fig:LN1}
\end{figure}

\begin{figure}
\centering
\includegraphics[width=4.4in]{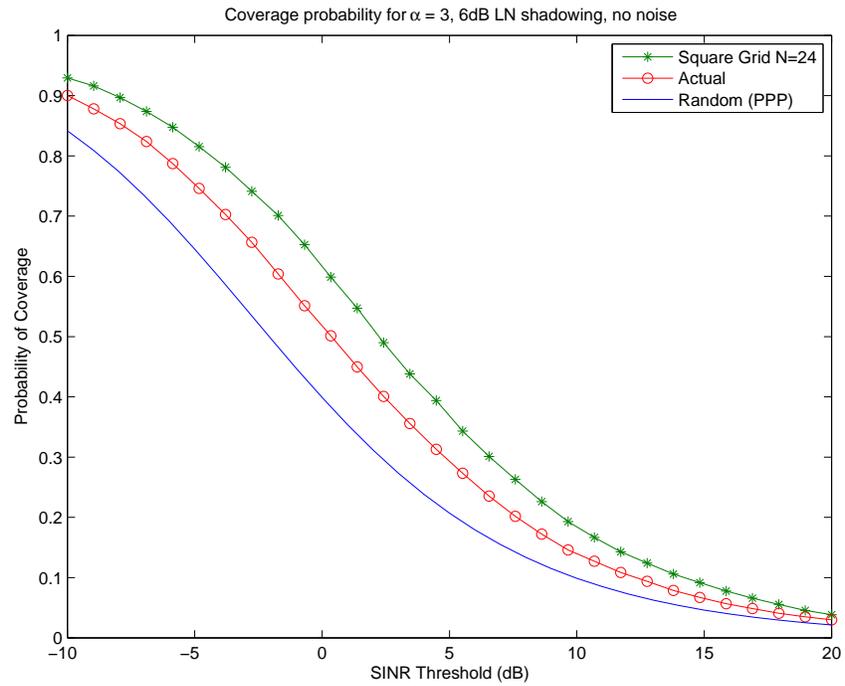}
\caption{Poisson vs. actual vs. grid base stations for $\alpha = 3$ with LN interference of 6 dB.}
\label{fig:LN2}
\end{figure}




\begin{figure}
\begin{center}
\includegraphics[width=3in]{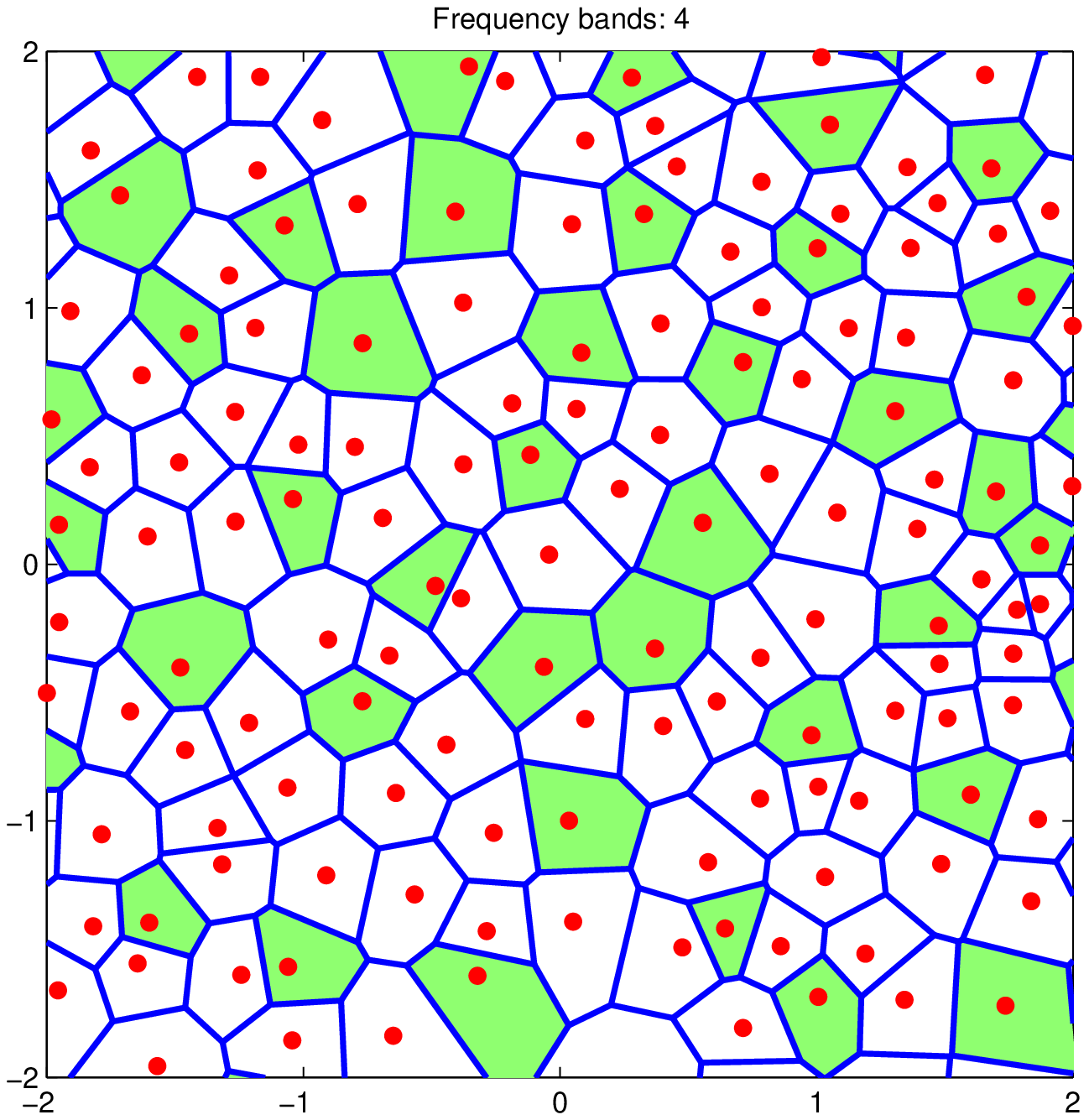}\includegraphics[width=3in]{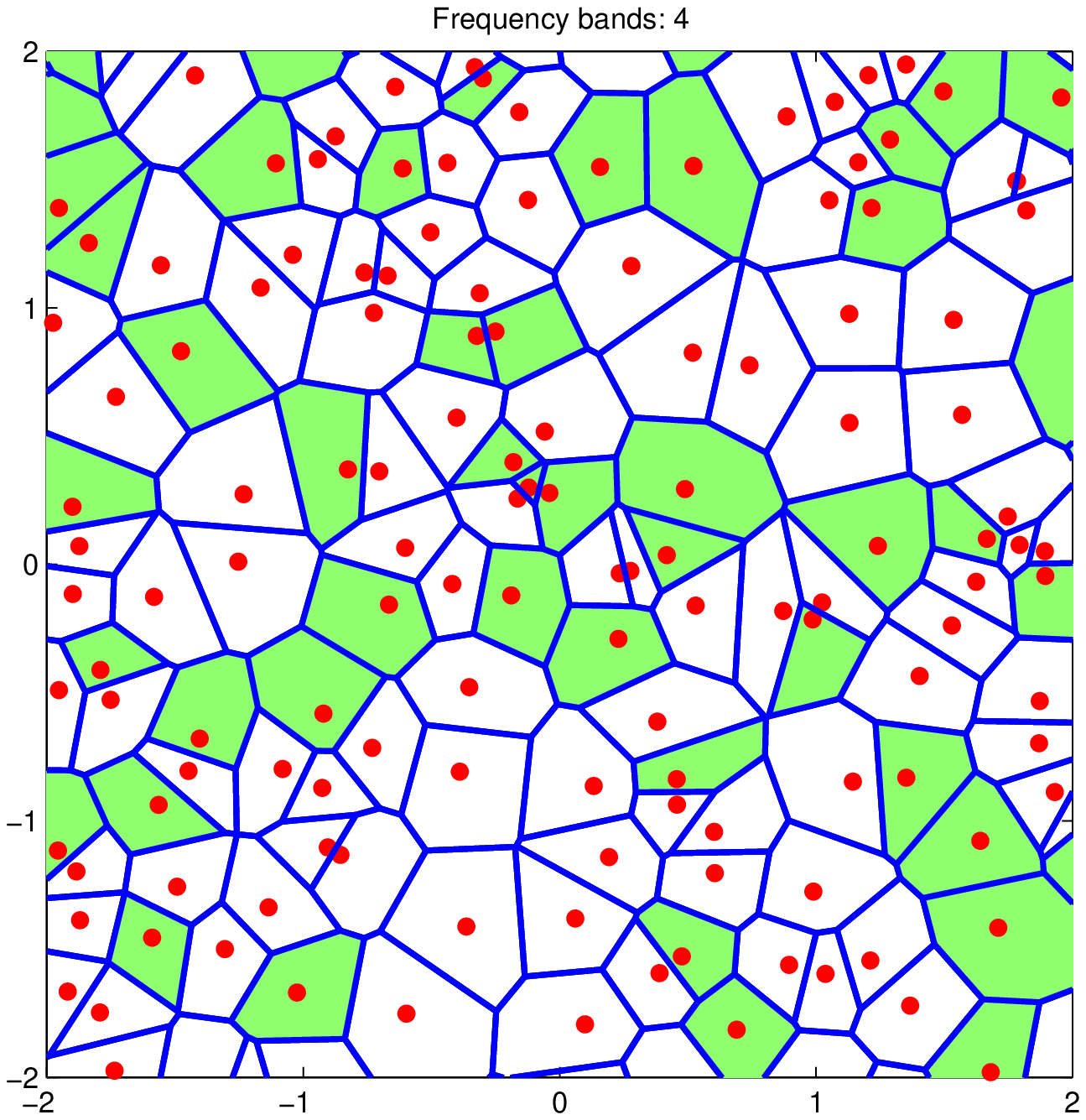}
\end{center}
\caption{ Left: A spatial reuse of $\delta=4$ is shown for an actual BS development  using a  greedy frequency  allocation. Right: Corresponds to a $\delta=4$ reuse for a  Poisson base station network with  random frequency allocation. The shaded cells use the same  frequency.}
\label{fig:freqReuse}
  \end{figure}

\begin{figure}
\centering
\includegraphics[width=3.2in]{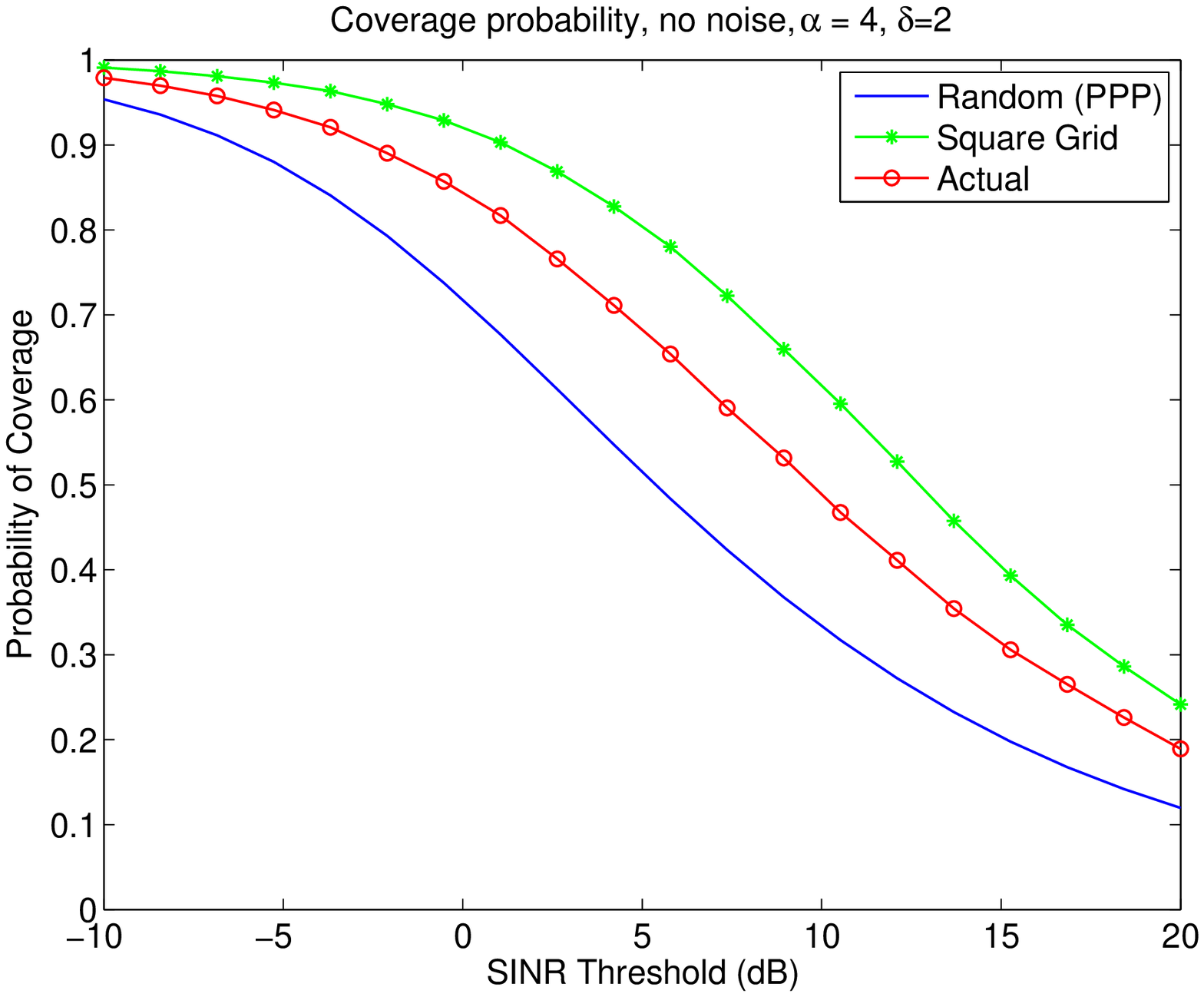}
\includegraphics[width=3.2in]{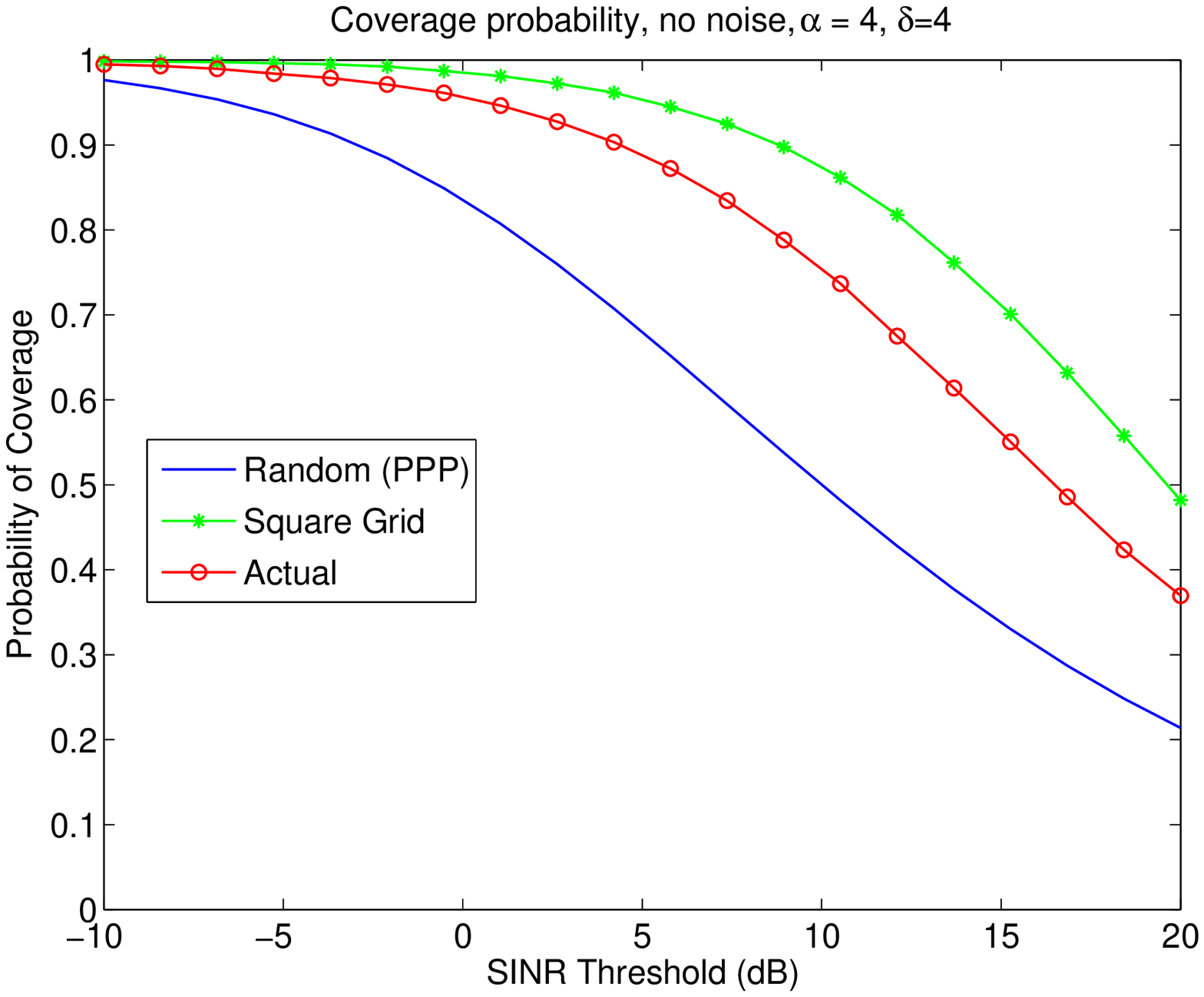}
\caption{Probability of coverage for frequency reuse factors $\delta =2$ (left) and $\delta =4$ (right). Lower spatial reuse (higher $\delta$) leads to better outage performance, and we observe that all 3 curves exhibit similar behavior.}
\label{fig:coverage_freq}
\end{figure}

\begin{figure}
\centering
\includegraphics[width=4.4in]{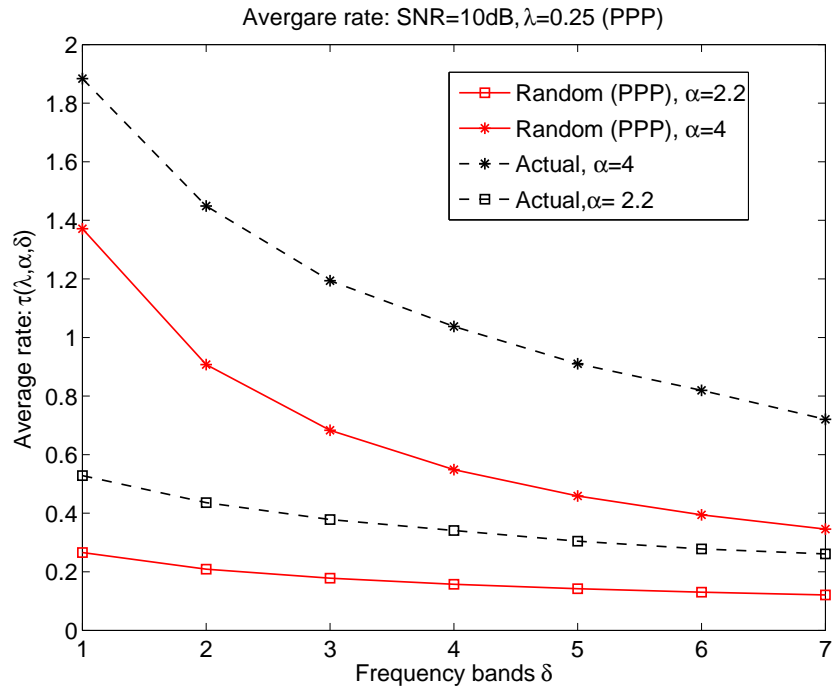}
\caption{Average rate of a typical user with $\snr=10$dB for both Poisson-distributed and actual base station locations.  The average rate is maximized when all the cells use the same frequency and hence the complete bandwidth.}
\label{fig:rate_freq}
\end{figure}


\end{document}